\documentclass[11pt]{article}

\usepackage[T1]{fontenc}
\usepackage[utf8]{inputenc}
\usepackage{lmodern}
\usepackage{microtype}
\usepackage[a4paper,margin=1in]{geometry}

\usepackage{amsmath,amssymb,amsthm,mathtools}
\usepackage{mathrsfs}
\usepackage{enumitem}
\usepackage{cite}
\usepackage{hyperref}

\hypersetup{
  colorlinks=true,
  linkcolor=blue,
  citecolor=blue,
  urlcolor=blue
}

\allowdisplaybreaks

\makeatletter
\renewcommand\paragraph{\@startsection{paragraph}{4}{\z@}%
  {1.25ex \@plus .5ex \@minus .2ex}%
  {.6ex \@plus .1ex}%
  {\normalfont\normalsize\bfseries}}
\makeatother

\newtheorem{theorem}{Theorem}[section]
\newtheorem{lemma}[theorem]{Lemma}
\newtheorem{proposition}[theorem]{Proposition}
\newtheorem{corollary}[theorem]{Corollary}
\newtheorem{fact}[theorem]{Fact}
\newtheorem{definition}[theorem]{Definition}
\newtheorem{remark}[theorem]{Remark}

\newtheorem{problem}[theorem]{Problem}

\newcommand{\FO}{\mathrm{FO}}

\newcommand{\qr}{\operatorname{qr}}
\newcommand{\Lang}{\mathcal{L}}
\newcommand{\N}{\mathbb{N}}
\newcommand{\eps}{\varepsilon}
\newcommand{\rankprof}{\rho}
\newcommand{\sepprof}{\sigma}
\newcommand{\qdist}{\delta_{\FO}}
\newcommand{\Syn}{\operatorname{Syn}}
\newcommand{\Dist}{\operatorname{Dist}}
\newcommand{\Mod}{\operatorname{Mod}}
\newcommand{\first}{\operatorname{first}}
\newcommand{\last}{\operatorname{last}}
\newcommand{\succr}{\operatorname{succ}}
\newcommand{\defect}{\operatorname{Def}}
\newcommand{\Ball}{\mathsf B}
\newcommand{\Ker}{\operatorname{Ker}}
\newcommand{\Min}{\operatorname{Min}}
\newcommand{\calP}{\mathcal P}
\newcommand{\calF}{\mathcal F}
\newcommand{\powind}{\operatorname{ind}}
\newcommand{\powper}{\operatorname{per}}

\title{Finite-Horizon First-Order Rank Profiles of Regular Languages}

\author{
Madina Bazarova\\
\small \href{mailto:madina.bazarova@bahcesehir.edu.tr}{madina.bazarova@bahcesehir.edu.tr}\\
\small Department of Computer Engineering\\
\small Bahcesehir University
\and
Faruk Alpay\\
\small \href{mailto:faruk.alpay@bahcesehir.edu.tr}{faruk.alpay@bahcesehir.edu.tr}\\
\small Department of Computer Engineering\\
\small Bahcesehir University
}

\date{April 29, 2026}

\begin{document}
\maketitle

\begin{abstract}
We introduce the finite-horizon first-order rank profile of a language
$L\subseteq\Sigma^\ast$: the least quantifier rank needed by an
$\FO[<]$ sentence to classify membership in $L$ correctly on all words of
length at most $n$. The invariant measures quantifier depth only; formula
size is deliberately not bounded. First, we prove a rank calculus that is
independent of regularity. Every language satisfies
$\rho_L(n)\le \lceil\log_2 n\rceil+4$, via balanced first-order distance
formulas and exact-word definitions. Moreover,
$\sup_n\rho_L(n)<\infty$ holds exactly when $L$ is globally
$\FO[<]$-definable, and the supremum equals the minimum quantifier rank of
such a definition. Second, for regular languages we prove a sharp
aperiodicity gap: if the syntactic monoid of $L$ is aperiodic, then
$\rho_L(n)=O(1)$; otherwise $\rho_L(n)=\log_2 n+O_L(1)$. The lower bound
extracts a nontrivial cyclic component from the syntactic monoid and
combines it with an Ehrenfeucht--Fra\"iss\'e power lemma for long
repetitions of a fixed word. Thus, for full $\FO[<]$ quantifier rank,
regular languages admit no intermediate finite-horizon growth between
bounded and logarithmic rank.
\end{abstract}

\section{Introduction}

A word
\[
w=a_1a_2\cdots a_m\in\Sigma^\ast
\]
is viewed as the finite ordered structure
\[
\mathbf w=(\{1,\dots,m\},<,(P_a)_{a\in\Sigma}),
\]
where
\[
P_a(i)
\quad\Longleftrightarrow\quad
a_i=a.
\]
Over this signature, monadic second-order logic defines exactly the regular
languages, while first-order logic with order, $\FO[<]$, defines exactly the
star-free regular languages. Algebraically, a regular language is first-order
definable if and only if its syntactic monoid is aperiodic
\cite{buchi1960,schutzenberger1965,mcnaughtonpapert1971,eilenberg1976,pin1986,straubing1994,thomas1997}.

This paper studies a finite-horizon quantitative version of that classical
picture. Given a language $L\subseteq\Sigma^\ast$ and a length bound $n$, we
ignore all words longer than $n$ and ask for the least first-order quantifier
rank needed to classify membership in $L$ on the finite ball
\[
\Sigma^{\le n}:=\{w\in\Sigma^\ast: |w|\le n\}.
\]
This least rank is denoted by $\rankprof_L(n)$.

The invariant deliberately records only quantifier rank. It does not impose
a bound on formula size. This distinction is essential. Every finite horizon
can be classified at logarithmic quantifier rank, even for arbitrary
nonregular languages, by writing exact-word formulas and disjoining those
corresponding to the accepted words of length at most $n$. The resulting
formula may be exponentially large. The question is therefore not succinctness
but the minimum nesting depth of first-order quantification required at a
given horizon.

\paragraph{Positioning of the contribution.}

The classical Sch\"utzenberger--McNaughton--Papert theorem characterizes
the regular languages definable in $\FO[<]$ as the star-free languages, or
equivalently as those with aperiodic syntactic monoid. This paper does not
claim a new global characterization of first-order definability. Its
contribution is quantitative and finite-horizon: it studies the least
first-order quantifier rank required to classify a language on the ball
$\Sigma^{\le n}$ and determines the possible asymptotic behaviour of this
rank for regular languages.

\paragraph{Scope of the main theorem.}

The main dichotomy is a statement about full $\FO[<]$ and quantifier rank
alone. It is not a succinctness theorem, since the formulas used for the
universal upper bound may have size exponential in the horizon. It also
does not automatically transfer to smaller fragments such as $\FO^2[<]$,
alternation levels, or positive first-order logic. These refinements are
therefore treated only as separate directions after the main theorem.

\paragraph{Main result.}

The profile has two general properties. First, for arbitrary languages,
boundedness is exactly global first-order definability:
\[
\sup_n\rankprof_L(n)<\infty
\quad\Longleftrightarrow\quad
L\in\FO[<].
\]
Second, for regular languages, the unbounded side is rigid:
\[
\Syn(L)\text{ nonaperiodic}
\quad\Longrightarrow\quad
\rankprof_L(n)=\log_2 n+O_L(1).
\]
Together with the Sch\"utzenberger--McNaughton--Papert theorem, this yields
the dichotomy
\[
\rankprof_L(n)=O(1)
\quad\text{or}\quad
\rankprof_L(n)=\log_2 n+O_L(1),
\]
with the bounded case occurring precisely for star-free languages.

\paragraph{Proof architecture.}

The proof separates the model-theoretic and algebraic components.

First, finite-horizon classification is expressed as saturation under the
rank-$q$ equivalence relation $\equiv_q$. This gives a quotient calculus valid
for all languages. Second, balanced first-order distance formulas yield
exact-word definitions of logarithmic quantifier rank, giving
\[
\rankprof_L(n)\le \lceil\log_2 n\rceil+4
\]
for every language $L$. Third, boundedness of $\rankprof_L$ is shown to be
equivalent to global saturation under some $\equiv_q$, hence to global
$\FO[<]$ definability. Finally, for regular non-star-free languages,
nonaperiodicity of the syntactic monoid yields contexted periodic powers
\[
r x^k s
\]
whose membership is eventually nonconstant modulo a period. An
Ehrenfeucht--Fra\"iss\'e power lemma shows that
\[
m,m'\ge 2^q
\quad\Longrightarrow\quad
x^m\equiv_q x^{m'}.
\]
Thus, at horizon $n$, one can choose two words of length at most $n$ with
different membership but identical rank-$q$ type for
\[
q\le \log_2 n-O_L(1).
\]
This gives the lower bound matching the universal upper bound up to an
additive constant.

\section{Preliminaries}

Fix a finite alphabet $\Sigma$. A word
\[
w=a_1a_2\cdots a_m\in\Sigma^\ast
\]
is represented as the finite structure
\[
\mathbf w=(\{1,\dots,m\},<,(P_a)_{a\in\Sigma}),
\]
where $<$ is the usual strict order and
\[
P_a(i)\quad\Longleftrightarrow\quad a_i=a.
\]
The empty word is represented by the empty structure.

We work with first-order logic over this signature and write $\FO[<]$ for it.
The quantifier rank of a formula $\varphi$ is denoted $\qr(\varphi)$.
For words $u,v$ and $q\in\N$, write
\[
u\equiv_q v
\]
if $u$ and $v$ satisfy the same $\FO[<]$ sentences of quantifier rank at most
$q$. Equivalently, Duplicator wins the $q$-round Ehrenfeucht--Fra\"iss\'e game
on $u$ and $v$ \cite{libkin2004}.

\subsection{Rank types}

\begin{fact}[Rank-$q$ characteristic formulas]\label{fact:types}
For every finite alphabet $\Sigma$, every $q\in\N$, and every word
$u\in\Sigma^\ast$, there exists an $\FO[<]$ sentence $\tau_{u,q}$ of
quantifier rank at most $q$ such that, for every word $v$,
\[
v\models\tau_{u,q}
\quad\Longleftrightarrow\quad
v\equiv_q u.
\]
Consequently, $\equiv_q$ has finite index on $\Sigma^\ast$.
\end{fact}

\begin{proof}
This is the standard Hintikka formula construction for finite relational
vocabularies. The finitely many rank-$q$ sentence types are obtained by
induction on $q$, using the finite vocabulary and Boolean closure. See
\cite[Chapter~3]{libkin2004}.
\end{proof}

\begin{lemma}[Composition]\label{lem:composition}
For every $q$, the relation $\equiv_q$ is a congruence on $\Sigma^\ast$.
Hence the quotient
\[
T_q(\Sigma):=\Sigma^\ast/{\equiv_q}
\]
is a finite monoid.
\end{lemma}

\begin{proof}
Suppose
\[
u\equiv_q u',
\qquad
v\equiv_q v'.
\]
We show that
\[
uv\equiv_q u'v'.
\]
In the $q$-round Ehrenfeucht--Fra\"iss\'e game on $uv$ and $u'v'$, Duplicator
combines the winning strategies on $(u,u')$ and on $(v,v')$. If Spoiler plays
in the left component of one concatenation, Duplicator answers in the
corresponding left component using the strategy for $(u,u')$. If Spoiler plays
in the right component, Duplicator uses the strategy for $(v,v')$.

After each round, the pebbled positions in the left components form a partial
isomorphism, and the pebbled positions in the right components form a partial
isomorphism. Every left-component position is smaller than every
right-component position in both concatenations. Therefore equality, order,
and letter predicates are preserved. This gives a winning strategy for
Duplicator.
\end{proof}

\begin{fact}[Long finite orders]\label{fact:linearorders}
Let $I_m$ denote the pure finite linear order of size $m$. If
\[
m,m'\ge2^q,
\]
then
\[
I_m\equiv_q I_{m'}.
\]
The same statement holds for unary words $a^m,a^{m'}$.
\end{fact}

\begin{proof}
Duplicator maintains the standard interval invariant. After a partial play
with $r$ rounds left, the selected points cut the two orders into corresponding
intervals $J,J'$ satisfying
\[
\min\{|J|,2^r\}=\min\{|J'|,2^r\}
\]
for every corresponding pair. Initially this holds because both orders have
size at least $2^q$. A move inside an interval of clipped size less than
$2^r$ is copied at the same offset. A move inside an interval of clipped size
$2^r$ is answered so that the two left subintervals and the two right
subintervals have equal clipped sizes at threshold $2^{r-1}$. Induction on
$r$ proves the claim.
\end{proof}

\subsection{Regular languages and syntactic monoids}

For a regular language $L\subseteq\Sigma^\ast$, its syntactic congruence is
\[
u\equiv_L v
\quad\Longleftrightarrow\quad
\forall r,s\in\Sigma^\ast,\
r u s\in L\Longleftrightarrow r v s\in L.
\]
The quotient monoid is denoted $\Syn(L)$, and the quotient morphism is
\[
\alpha_L:\Sigma^\ast\to\Syn(L).
\]
There is an accepting set $F_L\subseteq\Syn(L)$ such that
\[
L=\alpha_L^{-1}(F_L).
\]

A finite monoid $M$ is aperiodic if there exists $\omega\ge1$ such that
\[
m^\omega=m^{\omega+1}
\qquad(m\in M).
\]
Equivalently, every subgroup of $M$ is trivial. The
Sch\"utzenberger--McNaughton--Papert theorem states that, for regular $L$,
\[
L\in\FO[<]
\quad\Longleftrightarrow\quad
L\text{ is star-free}
\quad\Longleftrightarrow\quad
\Syn(L)\text{ is aperiodic}.
\]

\section{Finite-horizon rank profiles}

\begin{definition}[Rank profile]\label{def:rank-profile}
For $L\subseteq\Sigma^\ast$ and $n\in\N$, define
\[
\rankprof_L(n)
=
\min\Bigl\{
q\in\N:
\forall u,v\in\Sigma^{\le n},\
(u\in L,\ v\notin L)
\Rightarrow
u\not\equiv_q v
\Bigr\}.
\]
If $L\cap\Sigma^{\le n}$ is empty or all of $\Sigma^{\le n}$, then
$\rankprof_L(n)=0$.
\end{definition}

For distinct words $u,v$, define their first-order rank distance by
\[
\qdist(u,v)
=
\min\{q\in\N:u\not\equiv_q v\}.
\]

\begin{proposition}[Maximal obstruction]\label{prop:max-obstruction}
For every $L\subseteq\Sigma^\ast$ and every $n\in\N$,
\[
\rankprof_L(n)
=
\max
\Bigl\{
\qdist(u,v):
u,v\in\Sigma^{\le n},\
u\in L,\ v\notin L
\Bigr\},
\]
where the maximum of the empty set is $0$.
\end{proposition}

\begin{proof}
A rank $q$ separates all accepted/rejected pairs inside $\Sigma^{\le n}$
exactly when
\[
q\ge \qdist(u,v)
\]
for every such pair $(u,v)$. The least such $q$ is the displayed maximum.
\end{proof}

\begin{proposition}[Finite-horizon saturation]\label{prop:saturation}
For $L\subseteq\Sigma^\ast$, $n\in\N$, and $q\in\N$, the following are
equivalent:
\begin{enumerate}[label=\textup{(\roman*)},leftmargin=2.4em]
\item $\rankprof_L(n)\le q$;
\item $L\cap\Sigma^{\le n}$ is saturated under $\equiv_q$ inside
$\Sigma^{\le n}$:
\[
u\equiv_q v,\quad |u|,|v|\le n
\quad\Longrightarrow\quad
(u\in L\Longleftrightarrow v\in L);
\]
\item there exists an $\FO[<]$ sentence $\varphi$ with $\qr(\varphi)\le q$
such that
\[
\forall w\in\Sigma^{\le n},
\qquad
w\models\varphi
\Longleftrightarrow
w\in L.
\]
\end{enumerate}
\end{proposition}

\begin{proof}
The equivalence of (i) and (ii) is immediate from
Definition~\ref{def:rank-profile}.

Assume (ii). If $L\cap\Sigma^{\le n}$ is empty, a contradiction of rank $0$
may be used; if $L\cap\Sigma^{\le n}=\Sigma^{\le n}$, a tautology of rank $0$
may be used. Otherwise, let $\mathcal C_{q,n}(L)$ be the set of $\equiv_q$
classes $C$ satisfying
\[
C\cap\Sigma^{\le n}\neq\varnothing
\quad\text{and}\quad
C\cap\Sigma^{\le n}\subseteq L.
\]
For each $C\in\mathcal C_{q,n}(L)$ choose a representative $u_C$. By
Fact~\ref{fact:types},
\[
\varphi_{q,n,L}
=
\bigvee_{C\in\mathcal C_{q,n}(L)}
\tau_{u_C,q}
\]
has quantifier rank at most $q$. For every $w\in\Sigma^{\le n}$,
\[
w\models\varphi_{q,n,L}
\quad\Longleftrightarrow\quad
[w]_{\equiv_q}\in\mathcal C_{q,n}(L)
\quad\Longleftrightarrow\quad
w\in L.
\]
Thus (iii) holds.

Conversely, if $\varphi$ has rank at most $q$ and classifies $L$ on
$\Sigma^{\le n}$, then $\equiv_q$-equivalent words in $\Sigma^{\le n}$ agree
on $\varphi$, hence have the same membership in $L$.
\end{proof}

The profile is monotone:
\[
\rankprof_L(n)\le\rankprof_L(n+1),
\]
and complement invariant:
\[
\rankprof_L(n)=\rankprof_{\Sigma^\ast\setminus L}(n).
\]

\subsection{Automaton defect form}

Let $L$ be regular, recognized by
\[
\alpha:\Sigma^\ast\to M
\]
with accepting set $F\subseteq M$. For $q,n\in\N$, define the finite-horizon
defect
\[
\defect_{q,n}^{\alpha,F}
=
\left\{
(m_0,m_1)\in F\times(M\setminus F):
\begin{array}{l}
\exists u,v\in\Sigma^{\le n}\text{ such that}\\
\alpha(u)=m_0,\ \alpha(v)=m_1,\ u\equiv_q v
\end{array}
\right\}.
\]
Then
\[
\rankprof_L(n)
=
\min\{q:\defect_{q,n}^{\alpha,F}=\varnothing\}.
\]
Equivalently,
\[
\rankprof_L(n)\le q
\quad\Longleftrightarrow\quad
\Ker(\pi_q)\cap(\Sigma^{\le n}\times\Sigma^{\le n})
\subseteq
\Ker(\mathbf 1_F\circ\alpha),
\]
where
\[
\mathbf 1_F(m)=
\begin{cases}
1,&m\in F,\\
0,&m\notin F.
\end{cases}
\]
Thus $\rankprof_L(n)$ is the first level $q$ at which rank-$q$ first-order
types separate the accepting and rejecting fibers of the recognizer on the
finite ball
\[
\Ball_n(\Sigma):=\Sigma^{\le n}.
\]

\section{Exact words at logarithmic rank}

This section proves the universal upper bound. Define
\[
\first(x):=\neg\exists y\,(y<x),
\qquad
\last(x):=\neg\exists y\,(x<y),
\]
and
\[
\succr(x,y)
:=
x<y\wedge\neg\exists z\,(x<z\wedge z<y).
\]

\begin{lemma}[Balanced distance formulas]\label{lem:distance}
For every $d\ge0$ there exists an $\FO[<]$ formula $\Dist_d(x,y)$ such that,
in every finite word structure, $\Dist_d(x,y)$ holds exactly when $y$ is $d$
successor-steps to the right of $x$. Moreover,
\[
\qr(\Dist_d)\le
\begin{cases}
0,&d=0,\\[1mm]
\lceil\log_2 d\rceil+1,&d\ge1.
\end{cases}
\]
\end{lemma}

\begin{proof}
Set
\[
\Dist_0(x,y):=(x=y),
\qquad
\Dist_1(x,y):=\succr(x,y).
\]
For $d\ge2$, let
\[
d^-:=\lfloor d/2\rfloor,
\qquad
d^+:=\lceil d/2\rceil,
\]
and define
\[
\Dist_d(x,y)
:=
\exists z\,
\bigl(
\Dist_{d^-}(x,z)
\wedge
\Dist_{d^+}(z,y)
\bigr).
\]
Correctness follows by induction on $d$: a path of $d$ successor-steps from
$x$ to $y$ has a unique midpoint $z$ at distance $d^-$ from $x$ and distance
$d^+$ from $z$ to $y$.

Let $r(d):=\qr(\Dist_d)$. Then
\[
r(0)=0,\qquad r(1)=1,
\]
and, for $d\ge2$,
\[
r(d)\le1+\max\{r(d^-),r(d^+)\}.
\]
We prove
\[
r(d)\le\lceil\log_2 d\rceil+1
\]
by induction on $d\ge1$. The case $d=1$ is immediate. Let $d\ge2$ and put
$k=\lceil\log_2 d\rceil$. Then
\[
2^{k-1}<d\le2^k,
\]
and hence
\[
d^+=\lceil d/2\rceil\le2^{k-1}.
\]
If $d^+=1$, then $r(d)\le2\le k+1$. If $d^+\ge2$, the induction hypothesis
gives
\[
r(d)
\le
1+\bigl(\lceil\log_2 d^+\rceil+1\bigr)
\le
1+((k-1)+1)
=
k+1.
\]
This proves the bound.
\end{proof}

\begin{lemma}[Exact length]\label{lem:length}
For every $m\ge0$ there exists an $\FO[<]$ sentence $\lambda_m$ defining
words of length exactly $m$, with
\[
\qr(\lambda_0)=1
\]
and, for $m\ge1$,
\[
\qr(\lambda_m)\le \lceil\log_2 m\rceil+3.
\]
\end{lemma}

\begin{proof}
For $m=0$, take
\[
\lambda_0:=\neg\exists x\,(x=x).
\]
For $m\ge1$, define
\[
\lambda_m
:=
\exists f\,\exists \ell\,
\bigl(
\first(f)\wedge
\last(\ell)\wedge
\Dist_{m-1}(f,\ell)
\bigr).
\]
The formula says that the first and last positions are exactly $m-1$
successor-steps apart, hence the domain has exactly $m$ positions.

The two outer quantifiers contribute $2$. The deepest internal component has
rank at most
\[
\max\{1,\qr(\Dist_{m-1})\}.
\]
For $m=1$, this gives at most $3$. For $m\ge2$, Lemma~\ref{lem:distance}
gives
\[
\qr(\lambda_m)
\le
2+\max\{1,\lceil\log_2(m-1)\rceil+1\}
\le
\lceil\log_2 m\rceil+3.
\]
\end{proof}

\begin{lemma}[Exact-word formulas]\label{lem:exact-word}
For every word $w=a_1\cdots a_m\in\Sigma^\ast$, there exists an $\FO[<]$
sentence $\chi_w$ such that
\[
v\models\chi_w
\quad\Longleftrightarrow\quad
v=w
\]
for every word $v$. Moreover,
\[
\qr(\chi_w)\le
\begin{cases}
1,&m=0,\\[1mm]
\lceil\log_2 m\rceil+4,&m\ge1.
\end{cases}
\]
\end{lemma}

\begin{proof}
For $m=0$, set
\[
\chi_\eps:=\lambda_0.
\]
Assume $m\ge1$. Define
\[
\chi_w
:=
\exists f\,\exists \ell
\Biggl(
\first(f)\wedge
\last(\ell)\wedge
\Dist_{m-1}(f,\ell)
\wedge
\bigwedge_{i=1}^{m}
\exists x_i\,
\bigl(
\Dist_{i-1}(f,x_i)\wedge P_{a_i}(x_i)
\bigr)
\Biggr).
\]
The first three conjuncts force the domain to have exactly $m$ positions.
The conjunct indexed by $i$ forces the unique position at distance $i-1$
from the first position to satisfy $P_{a_i}$. Hence the word is exactly
$a_1\cdots a_m$.

For rank, the outer quantifiers over $f,\ell$ contribute $2$. Each inner
witness $x_i$ contributes one additional quantifier along that branch. The
largest distance used has parameter at most $m-1$. Thus, for $m\ge2$,
\[
\qr(\chi_w)
\le
3+\max_{0\le d\le m-1}\qr(\Dist_d)
\le
3+\lceil\log_2(m-1)\rceil+1
\le
\lceil\log_2 m\rceil+4.
\]
The case $m=1$ satisfies the same bound.
\end{proof}

\begin{theorem}[Universal logarithmic upper bound]\label{thm:upper}
For every finite alphabet $\Sigma$, every language $L\subseteq\Sigma^\ast$,
and every $n\ge1$,
\[
\rankprof_L(n)\le \lceil\log_2 n\rceil+4.
\]
\end{theorem}

\begin{proof}
For each $w\in L\cap\Sigma^{\le n}$, let $\chi_w$ be the exact-word sentence
from Lemma~\ref{lem:exact-word}. Define
\[
\varphi_{L,n}
:=
\bigvee_{w\in L\cap\Sigma^{\le n}}\chi_w.
\]
If $L\cap\Sigma^{\le n}=\varnothing$, take $\varphi_{L,n}$ to be a
contradiction of rank $0$. Otherwise, for every $v\in\Sigma^{\le n}$,
\[
v\models\varphi_{L,n}
\quad\Longleftrightarrow\quad
v\in L.
\]
Every nonempty disjunct has quantifier rank at most
$\lceil\log_2 n\rceil+4$, and the empty-word disjunct has rank $1$. Hence the
whole disjunction has quantifier rank at most $\lceil\log_2 n\rceil+4$.
Proposition~\ref{prop:saturation} gives the claim.
\end{proof}

\begin{remark}[Rank versus size]\label{rem:rank-versus-size}
Theorem~\ref{thm:upper} is purely about rank. With formulas written as trees,
the construction may have size exponential in $n$. Indeed,
\[
|\{w\in\Sigma^\ast:|w|\le n\}|
=
\sum_{m=0}^{n}|\Sigma|^m
=
\begin{cases}
n+1,&|\Sigma|=1,\\[1mm]
\dfrac{|\Sigma|^{n+1}-1}{|\Sigma|-1},&|\Sigma|\ge2.
\end{cases}
\]
Thus the disjunction in $\varphi_{L,n}$ may contain exponentially many
exact-word components.
\end{remark}

\section{Bounded profiles and separator profiles}

\begin{theorem}[Bounded profile theorem]\label{thm:bounded}
For every language $L\subseteq\Sigma^\ast$,
\[
\sup_{n\in\N}\rankprof_L(n)<\infty
\]
if and only if $L$ is definable in $\FO[<]$. More precisely,
\[
\sup_{n\in\N}\rankprof_L(n)
=
\min\{\qr(\varphi):\varphi\in\FO[<]\text{ defines }L\},
\]
with value $+\infty$ if no such $\varphi$ exists.
\end{theorem}

\begin{proof}
If $\varphi$ defines $L$ and $\qr(\varphi)=q$, then $\varphi$ classifies $L$
correctly on every $\Sigma^{\le n}$. Hence
\[
\rankprof_L(n)\le q
\qquad(n\in\N).
\]
Thus
\[
\sup_n\rankprof_L(n)
\le
\min\{\qr(\varphi):\varphi\text{ defines }L\}.
\]

Conversely, suppose
\[
\sup_n\rankprof_L(n)\le q.
\]
We show that $L$ is saturated under $\equiv_q$ on all of $\Sigma^\ast$.
Let
\[
u\equiv_q v.
\]
Put
\[
n=\max\{|u|,|v|\}.
\]
If $u$ and $v$ had different membership in $L$, the definition of
$\rankprof_L(n)$ would imply
\[
u\not\equiv_q v,
\]
a contradiction. Hence
\[
u\equiv_q v
\quad\Longrightarrow\quad
(u\in L\Longleftrightarrow v\in L).
\]
Therefore $L$ is a union of $\equiv_q$-classes. Since $\equiv_q$ has finite
index by Fact~\ref{fact:types}, $L$ is defined by
\[
\bigvee_{C\subseteq L}\tau_{u_C,q},
\]
where $C$ ranges over the $\equiv_q$-classes contained in $L$ and $u_C$ is a
representative of $C$. This sentence has quantifier rank at most $q$.

Taking the least possible $q$ gives the equality.
\end{proof}

\begin{definition}[Separator profile]\label{def:separator-profile}
Let $K,H\subseteq\Sigma^\ast$ be disjoint. Define
\[
\sepprof_{K,H}(n)
=
\min\Bigl\{
q\in\N:
\forall u\in K\cap\Sigma^{\le n},
\forall v\in H\cap\Sigma^{\le n},
u\not\equiv_q v
\Bigr\}.
\]
\end{definition}

\begin{theorem}[Bounded separator theorem]\label{thm:separator}
Let $K,H\subseteq\Sigma^\ast$ be disjoint. Then
\[
\sup_n\sepprof_{K,H}(n)<\infty
\]
if and only if there exists an $\FO[<]$ sentence $\psi$ such that
\[
K\subseteq\Lang(\psi),
\qquad
H\cap\Lang(\psi)=\varnothing.
\]
Moreover,
\[
\sup_n\sepprof_{K,H}(n)
=
\min\{\qr(\psi):\psi\text{ is an }\FO[<]\text{ separator of }K\text{ and }H\},
\]
with value $+\infty$ if no separator exists.
\end{theorem}

\begin{proof}
If $\psi$ separates $K$ from $H$ and $\qr(\psi)=q$, then no word of $K$ and
no word of $H$ can be $\equiv_q$-equivalent. Hence
\[
\sepprof_{K,H}(n)\le q
\qquad(n\in\N).
\]

Conversely, suppose
\[
\sup_n\sepprof_{K,H}(n)\le q.
\]
No $\equiv_q$-class can meet both $K$ and $H$. Indeed, if $u\in K$,
$v\in H$, and $u\equiv_q v$, then for
\[
n=\max\{|u|,|v|\}
\]
we contradict the definition of $\sepprof_{K,H}(n)$.

Let $\mathcal C$ be the set of $\equiv_q$-classes meeting $K$. Then
\[
\psi
=
\bigvee_{C\in\mathcal C}\tau_{u_C,q}
\]
contains $K$ and avoids $H$. The rank is at most $q$. Taking the least
possible $q$ gives the equality.
\end{proof}

For regular pairs, first-order separability is decidable and has algebraic
characterizations \cite{placezeitoun2016,placezeitoun2019}. Theorem
\ref{thm:separator} is not a decidability statement; it identifies exactly
what bounded finite-horizon separator rank means.

\section{Powers and first-order rank}

The regular lower bound requires a rank estimate for powers of a fixed
nonempty word.

\begin{lemma}[Block-power equivalence]\label{lem:block-power}
Let $x\in\Sigma^+$ be nonempty. If
\[
m,m'\ge2^q,
\]
then
\[
x^m\equiv_q x^{m'}.
\]
\end{lemma}

\begin{proof}
Write
\[
x=b_1b_2\cdots b_\ell,
\qquad
\ell=|x|\ge1.
\]
Every position in $x^m$ is uniquely represented as a pair
\[
(i,j)
\quad
(1\le i\le m,\ 1\le j\le\ell),
\]
where $i$ is the block index and $j$ is the offset inside the block. The order
is lexicographic:
\[
(i,j)<(i',j')
\quad\Longleftrightarrow\quad
i<i'
\text{ or }(i=i'\text{ and }j<j').
\]
The letter at $(i,j)$ is $b_j$, independent of $i$.

By Fact~\ref{fact:linearorders}, Duplicator wins the $q$-round game on $I_m$
and $I_{m'}$. In the game on $x^m$ and $x^{m'}$, Duplicator follows that
strategy on block indices and copies offsets. If Spoiler plays $(i,j)$,
Duplicator answers with $(i',j)$, where $i'$ is the response to $i$ in the
pure-order game.

For any two pebbled positions $(i,j)$ and $(k,\ell')$, the order relation is
determined by the order and equality relation between $i$ and $k$, together
with the fixed offsets $j$ and $\ell'$. The pure-order strategy preserves
order and equality of block indices, while the offsets are copied exactly.
Therefore equality, order, and letter predicates are preserved among pebbled
positions. Duplicator wins the $q$-round game on $x^m,x^{m'}$.
\end{proof}

\begin{corollary}[Contexted power equivalence]\label{cor:context-power}
Let
\[
r,s\in\Sigma^\ast,
\qquad
x\in\Sigma^+.
\]
If
\[
m,m'\ge2^q,
\]
then
\[
r x^m s\equiv_q r x^{m'}s.
\]
\end{corollary}

\begin{proof}
By Lemma~\ref{lem:block-power},
\[
x^m\equiv_q x^{m'}.
\]
By Lemma~\ref{lem:composition}, $\equiv_q$ is a congruence. Hence
\[
r x^m s\equiv_q r x^{m'}s.
\]
\end{proof}

\section{The regular aperiodicity gap}

We now prove the main theorem. The algebraic part is separated from the
first-order part.

\begin{lemma}[Syntactic cycle extraction]\label{lem:syntactic-extraction}
Let $L\subseteq\Sigma^\ast$ be regular and suppose $\Syn(L)$ is not
aperiodic. Then there exist
\[
r,s\in\Sigma^\ast,
\qquad
x\in\Sigma^+,
\]
integers
\[
h\ge0,
\qquad
p\ge2,
\]
and residues
\[
i,j\in\{0,\dots,p-1\}
\]
such that, after possibly interchanging the two membership values,
\[
r x^{h+i+tp}s\in L
\quad\text{and}\quad
r x^{h+j+tp}s\notin L
\]
for every $t\ge0$.
\end{lemma}

\begin{proof}
Let
\[
M=\Syn(L),
\qquad
\alpha=\alpha_L,
\qquad
F=F_L.
\]
Since $M$ is not aperiodic, there exists $a\in M$ whose power sequence has
eventual period $p\ge2$. Let $h=\powind(a)$ and $p=\powper(a)$. Then
\[
a^{k+p}=a^k
\qquad(k\ge h),
\]
and the cycle
\[
\Gamma(a):=\{a^h,a^{h+1},\dots,a^{h+p-1}\}
\]
has at least two elements.

Choose residues $i,j\in\{0,\dots,p-1\}$ with
\[
a^{h+i}\ne a^{h+j}.
\]
Because $M$ is the syntactic monoid of $L$, distinct elements are separated
by contexts. Hence there exist $m_\ell,m_r\in M$ such that
\[
m_\ell a^{h+i}m_r\in F
\quad\text{and}\quad
m_\ell a^{h+j}m_r\notin F,
\]
or the same statement with the membership values interchanged.

Since $\alpha$ is onto, choose words $r,s\in\Sigma^\ast$ such that
\[
\alpha(r)=m_\ell,
\qquad
\alpha(s)=m_r.
\]
Choose $x\in\Sigma^\ast$ with
\[
\alpha(x)=a.
\]
The empty word maps to the identity of $M$. Since $\powper(a)\ge2$, the
element $a$ is not the identity. Therefore no representative of $a$ is empty,
and we may take $x\in\Sigma^+$.

For every $t\ge0$,
\[
a^{h+i+tp}=a^{h+i},
\qquad
a^{h+j+tp}=a^{h+j}.
\]
Therefore
\[
\alpha(r x^{h+i+tp}s)
=
m_\ell a^{h+i}m_r,
\]
and
\[
\alpha(r x^{h+j+tp}s)
=
m_\ell a^{h+j}m_r.
\]
The desired membership distinction follows.
\end{proof}

\begin{lemma}[Logarithmic obstruction from a contexted cycle]\label{lem:cycle-lower}
Assume there are
\[
r,s\in\Sigma^\ast,
\quad
x\in\Sigma^+,
\quad
h\ge0,
\quad
p\ge2
\]
such that membership of $r x^k s$ in $L$ is nonconstant on the residue classes
of $k$ modulo $p$ for all $k\ge h$. Put
\[
C:=|r|+|s|,
\qquad
\ell:=|x|.
\]
For every $n$ such that
\[
K(n):=\left\lfloor\frac{n-C}{\ell}\right\rfloor
\]
satisfies
\[
K(n)-p+1\ge h
\quad\text{and}\quad
K(n)-p+1\ge1,
\]
one has
\[
\rankprof_L(n)
\ge
\left\lfloor
\log_2(K(n)-p+1)
\right\rfloor+1.
\]
\end{lemma}

\begin{proof}
The interval
\[
B_n:=\{K(n)-p+1,
K(n)-p+2,
\dots,K(n)\}
\]
has length $p$, lies above $h$, and contains one representative of every
residue class modulo $p$. By the nonconstant-residue assumption, there exist
\[
k,k'\in B_n
\]
such that
\[
r x^k s\in L,
\qquad
r x^{k'}s\notin L.
\]
Both words have length at most $n$, since
\[
|r x^k s|
=
C+\ell k
\le
C+\ell K(n)
\le n,
\]
and similarly for $k'$.

Let
\[
q=
\left\lfloor
\log_2(K(n)-p+1)
\right\rfloor.
\]
Then
\[
2^q\le K(n)-p+1\le k,k'.
\]
By Corollary~\ref{cor:context-power},
\[
r x^k s\equiv_q r x^{k'}s.
\]
Thus rank $q$ does not classify $L$ correctly on $\Sigma^{\le n}$. Hence
\[
\rankprof_L(n)
\ge q+1
=
\left\lfloor
\log_2(K(n)-p+1)
\right\rfloor+1.
\]
\end{proof}

\begin{theorem}[Finite-horizon aperiodicity gap]\label{thm:aperiodicity-gap}
Let $L\subseteq\Sigma^\ast$ be regular. Then exactly one of the following
holds:
\begin{enumerate}[label=\textup{(\roman*)},leftmargin=2.4em]
\item $\Syn(L)$ is aperiodic, equivalently $L\in\FO[<]$, and
\[
\rankprof_L(n)=O(1).
\]
\item $\Syn(L)$ is not aperiodic, and
\[
\rankprof_L(n)=\log_2 n+O_L(1).
\]
More explicitly, there are constants
\[
C\ge0,
\qquad
\ell\ge1,
\qquad
p\ge2,
\qquad
N_0\ge1
\]
such that, for every $n\ge N_0$,
\[
\rankprof_L(n)
\ge
\left\lfloor
\log_2
\left(
\left\lfloor\frac{n-C}{\ell}\right\rfloor-p+1
\right)
\right\rfloor+1.
\]
At the same time,
\[
\rankprof_L(n)
\le
\lceil\log_2 n\rceil+4
\qquad(n\ge1).
\]
\end{enumerate}
\end{theorem}

\begin{proof}
If $\Syn(L)$ is aperiodic, then $L$ is $\FO[<]$-definable by the
Sch\"utzenberger--McNaughton--Papert theorem. Theorem~\ref{thm:bounded}
gives
\[
\rankprof_L(n)=O(1).
\]

Assume $\Syn(L)$ is not aperiodic. By
Lemma~\ref{lem:syntactic-extraction}, there are
\[
r,s\in\Sigma^\ast,
\quad
x\in\Sigma^+,
\quad
h\ge0,
\quad
p\ge2
\]
such that membership of $r x^k s$ is eventually nonconstant as a function of
$k\bmod p$. Let
\[
C=|r|+|s|,
\qquad
\ell=|x|.
\]
Choose $N_0$ so large that, for all $n\ge N_0$,
\[
\left\lfloor\frac{n-C}{\ell}\right\rfloor-p+1\ge h
\quad\text{and}\quad
\left\lfloor\frac{n-C}{\ell}\right\rfloor-p+1\ge1.
\]
Lemma~\ref{lem:cycle-lower} gives the displayed lower bound.

The universal upper bound gives
\[
\rankprof_L(n)\le\lceil\log_2 n\rceil+4.
\]
Since
\[
\log_2
\left(
\left\lfloor\frac{n-C}{\ell}\right\rfloor-p+1
\right)
=
\log_2 n+O_L(1),
\]
we obtain
\[
\rankprof_L(n)=\log_2 n+O_L(1).
\]
\end{proof}

\begin{corollary}[No intermediate regular rank growth]\label{cor:no-intermediate}
If $L\subseteq\Sigma^\ast$ is regular, then
\[
\rankprof_L(n)=O(1)
\quad\text{or}\quad
\rankprof_L(n)=\log_2 n+O_L(1).
\]
The bounded case occurs exactly for star-free languages; the logarithmic case
occurs exactly for regular languages that are not star-free.
\end{corollary}

\begin{proof}
Combine Theorem~\ref{thm:aperiodicity-gap} with the star-free/aperiodic
characterization.
\end{proof}

\subsection{Quantitative syntactic constants}

The proof gives more than a dichotomy. It gives a lower-bound constant from
a syntactic witness. Define $\Lambda_L$ to be the set of quadruples
\[
(C,\ell,p,h)
\]
for which there exist
\[
r,s\in\Sigma^\ast,
\qquad
x\in\Sigma^+,
\qquad
|r|+|s|=C,
\qquad
|x|=\ell,
\]
such that membership of $r x^k s$ is nonconstant on residue classes modulo
$p$ for all $k\ge h$. For
\[
\lambda=(C,\ell,p,h)\in\Lambda_L,
\]
write
\[
B_\lambda(n)
=
\left\lfloor
\log_2
\left(
\left\lfloor\frac{n-C}{\ell}\right\rfloor-p+1
\right)
\right\rfloor+1.
\]
Then the proof gives
\[
\rankprof_L(n)\ge
\max_{\lambda\in\Lambda_L} B_\lambda(n)
\]
whenever the logarithms are defined. Since
\[
B_\lambda(n)=\log_2 n-\log_2\ell+O_\lambda(1),
\]
the leading coefficient is always one with respect to $\log_2 n$, but the
additive constant depends on the shortest available nonaperiodic witness.

\section{Unary regular languages}

The unary case is a minimal witness for the aperiodicity gap. A unary regular
language is ultimately periodic: if $L\subseteq\{a\}^\ast$ is regular, then
there exist
\[
N\ge0,\qquad p\ge1,\qquad R\subseteq\{0,\dots,p-1\}
\]
such that
\[
a^m\in L
\quad\Longleftrightarrow\quad
m\bmod p\in R
\qquad(m\ge N).
\]

\begin{proposition}[Unary first-order languages]\label{prop:unary-fo}
For $L\subseteq\{a\}^\ast$, the following are equivalent:
\begin{enumerate}[label=\textup{(\roman*)},leftmargin=2.4em]
\item $L$ is $\FO[<]$-definable;
\item $L$ is finite or cofinite.
\end{enumerate}
\end{proposition}

\begin{proof}
If $L$ is finite, say
\[
L=\{a^{m_1},\dots,a^{m_t}\},
\]
then
\[
L
=
\Lang(\chi_{a^{m_1}}\vee\cdots\vee\chi_{a^{m_t}})
\]
by Lemma~\ref{lem:exact-word}. Cofinite languages are complements of finite
languages, hence are also first-order definable.

Conversely, suppose $L$ is defined by a sentence of rank $q$. By
Fact~\ref{fact:linearorders},
\[
a^m\equiv_q a^{m'}
\qquad(m,m'\ge2^q).
\]
Therefore all words of length at least $2^q$ have the same membership in
$L$. Hence $L$ is finite or cofinite.
\end{proof}

\begin{corollary}[Unary regular dichotomy]\label{cor:unary-dichotomy}
Let $L\subseteq\{a\}^\ast$ be regular. Then exactly one of the following
holds:
\[
L\text{ is finite or cofinite}
\quad\text{and}\quad
\rankprof_L(n)=O(1),
\]
or
\[
L\text{ is neither finite nor cofinite}
\quad\text{and}\quad
\rankprof_L(n)=\log_2 n+O_L(1).
\]
\end{corollary}

\begin{proof}
By Proposition~\ref{prop:unary-fo}, unary first-order definability is
equivalent to being finite or cofinite. Apply
Theorem~\ref{thm:aperiodicity-gap}.
\end{proof}

For the modular language
\[
\Mod_{p,R}
=
\{a^m:m\bmod p\in R\},
\]
where
\[
p\ge2,
\qquad
\varnothing\ne R\subsetneq\{0,\dots,p-1\},
\]
the explicit bounds are
\[
\lfloor\log_2(n-p+1)\rfloor+1
\le
\rankprof_{\Mod_{p,R}}(n)
\le
\lceil\log_2 n\rceil+4
\qquad(n\ge p).
\]
For parity,
\[
\mathsf{Even}:=\{a^{2m}:m\ge0\},
\]
this gives
\[
\lfloor\log_2(n-1)\rfloor+1
\le
\rankprof_{\mathsf{Even}}(n)
\le
\lceil\log_2 n\rceil+4
\qquad(n\ge2).
\]

Threshold languages show the bounded side. For
\[
T_t=\{a^m:m\ge t\},
\]
the complement is
\[
\{a^0,a^1,\dots,a^{t-1}\}.
\]
Using Lemma~\ref{lem:exact-word},
\[
\sup_n\rankprof_{T_t}(n)
\le
\lceil\log_2 t\rceil+4
\qquad(t\ge1).
\]

\section{Further directions}

The aperiodicity gap is sharp for regular languages when quantifier rank is
the only measured resource. The following refinements are not part of the
main classification. They are included to delimit what the theorem proves and
what it does not prove.

\subsection{Rank-size Pareto profiles}

For a sentence $\varphi$, let $|\varphi|$ denote its formula-tree size. Define
the finite-horizon rank-size feasible set
\[
\calP_L(n)
=
\left\{
(q,s)\in\N^2:
\exists\varphi\in\FO[<]\text{ such that }
\qr(\varphi)\le q,\ |\varphi|\le s,\ 
\forall w\in\Sigma^{\le n},\
w\models\varphi\Longleftrightarrow w\in L
\right\}.
\]
The rank profile is the projection
\[
\rankprof_L(n)=\min\{q:\exists s,\ (q,s)\in\calP_L(n)\}.
\]
The size profile is the other projection:
\[
s_L(n)=\min\{s:\exists q,\ (q,s)\in\calP_L(n)\}.
\]

The exact-word construction gives a rank bound but not a small-size bound. If
the recursive distance formula is expanded as a tree, its size satisfies
\[
S(0)=O(1),
\qquad
S(1)=O(1),
\]
and
\[
S(d)\le S(\lfloor d/2\rfloor)+S(\lceil d/2\rceil)+O(1),
\]
so
\[
S(d)=O(d).
\]
For a word $w$ of length $m$, the formula $\chi_w$ contains distance formulas
with parameters
\[
0,1,\dots,m-1.
\]
Thus a direct expansion yields
\[
|\chi_w|
=
O\left(\sum_{i=0}^{m-1} S(i)\right)
=
O(m^2).
\]
Consequently, the naive finite-horizon classifier satisfies
\[
|\varphi_{L,n}|
\le
O\left(
\sum_{w\in L\cap\Sigma^{\le n}} |w|^2
\right)
\le
O\left(
n^2\sum_{m=0}^{n}|\Sigma|^m
\right).
\]
For $|\Sigma|\ge2$, this is
\[
|\varphi_{L,n}|=O(n^2|\Sigma|^n).
\]
Thus Theorem~\ref{thm:upper} should be read as
\[
(\lceil\log_2 n\rceil+4,\ O(n^2|\Sigma|^n))
\in\calP_L(n),
\]
not as a succinctness statement.

A more robust invariant is the Pareto frontier
\[
\Min\calP_L(n)
=
\left\{
(q,s)\in\calP_L(n):
\not\exists(q',s')\in\calP_L(n)
\text{ with }
q'\le q,\ s'\le s,\ (q',s')\ne(q,s)
\right\}.
\]
For regular non-star-free languages, Theorem~\ref{thm:aperiodicity-gap}
implies that every point
\[
(q,s)\in\calP_L(n)
\]
satisfies
\[
q\ge\log_2 n-O_L(1).
\]
It gives no lower bound on $s$ once $q$ is logarithmic.

\begin{problem}[Rank-size tradeoff for regular languages]\label{prob:rank-size}
For regular $L$, determine the asymptotic shape of
\[
\Min\calP_L(n).
\]
In particular, for non-star-free regular $L$, decide whether there are
constants $A,B$ such that
\[
(\lceil\log_2 n\rceil+A,\ n^B)\in\calP_L(n)
\]
for all $n$, and characterize the languages for which this is possible.
\end{problem}

\subsection{Fragment profiles}

Let $\calF\subseteq\FO[<]$ be a syntactically defined fragment. Define
\[
\rankprof_L^{\calF}(n)
=
\min\left\{
q:
\exists\varphi\in\calF,\
\qr(\varphi)\le q,\
\forall w\in\Sigma^{\le n},\
w\models\varphi\Longleftrightarrow w\in L
\right\},
\]
with value $+\infty$ if no such formula exists.

The proof of Theorem~\ref{thm:bounded} used two closure properties of full
first-order logic:
\[
\text{finite rank-}q\text{ type definability}
\]
and
\[
\text{finite disjunction of rank-}q\text{ type definitions}.
\]
For a fragment $\calF$, the analogue may fail. Even when one can define a
fragment equivalence
\[
u\equiv_q^{\calF} v
\quad\Longleftrightarrow\quad
u,v\text{ agree on all }\calF\text{-sentences of rank }\le q,
\]
it need not be true that every $\equiv_q^{\calF}$-class is itself definable
in $\calF$ at rank $q$. Therefore the implication
\[
\text{finite-horizon saturation}
\quad\Longrightarrow\quad
\text{finite-horizon definability}
\]
must be reproved fragment by fragment.

For a fragment $\calF$, distinguish three profiles:
\[
\rankprof_L^{\calF,\mathrm{sem}}(n)
=
\min\left\{
q:
u\equiv_q^{\calF}v,\ |u|,|v|\le n
\Rightarrow
(u\in L\Longleftrightarrow v\in L)
\right\},
\]
\[
\rankprof_L^{\calF,\mathrm{syn}}(n)
=
\min\left\{
q:
\exists\varphi\in\calF,\
\qr(\varphi)\le q,\
\varphi\text{ classifies }L\text{ on }\Sigma^{\le n}
\right\},
\]
and
\[
\rankprof_L^{\calF,\mathrm{bool}}(n)
=
\min\left\{
q:
L\cap\Sigma^{\le n}
\text{ is a Boolean combination of }\equiv_q^{\calF}\text{-classes}
\right\}.
\]
For full $\FO[<]$, these three coincide:
\[
\rankprof_L^{\FO,\mathrm{sem}}(n)
=
\rankprof_L^{\FO,\mathrm{syn}}(n)
=
\rankprof_L^{\FO,\mathrm{bool}}(n)
=
\rankprof_L(n).
\]
For fragments such as $\FO^2[<]$, alternation levels, or positive first-order
logic, the equalities become substantive questions.

\begin{problem}[Fragment gap problem]\label{prob:fragment-gap}
Let $\calF$ be one of
\[
\FO^2[<],
\qquad
\Sigma_i[<],
\qquad
\mathcal B\Sigma_i[<],
\qquad
\FO^+[<].
\]
Classify the regular languages $L$ for which
\[
\rankprof_L^{\calF,\mathrm{syn}}(n)=O(1),
\]
and determine whether unbounded profiles must be logarithmic or can have
intermediate growth.
\end{problem}

The proof of Theorem~\ref{thm:aperiodicity-gap} does not immediately transfer
to these fragments. The syntactic extraction still produces contexted powers
\[
r x^k s,
\]
but the power lemma
\[
x^m\equiv_q x^{m'}
\qquad(m,m'\ge2^q)
\]
is a full $\FO[<]$ statement. A fragment may distinguish powers with a
different threshold, or may fail to express the Boolean combinations needed
to isolate rank-$q$ types.

\subsection{Separator growth beyond boundedness}

Theorem~\ref{thm:separator} gives the exact boundedness criterion:
\[
\sup_n\sepprof_{K,H}(n)<\infty
\quad\Longleftrightarrow\quad
K,H\text{ are }\FO[<]\text{-separable}.
\]
For regular pairs, the profile question is finer.

Let $K,H$ be regular, recognized by a common morphism
\[
\alpha:\Sigma^\ast\to M
\]
with accepting sets
\[
F_K,F_H\subseteq M,
\qquad
F_K\cap F_H=\varnothing.
\]
Define the separator defect
\[
\defect_{q,n}^{K,H}
=
\left\{
(m_K,m_H)\in F_K\times F_H:
\begin{array}{l}
\exists u\in K\cap\Sigma^{\le n},
\exists v\in H\cap\Sigma^{\le n}\\
\alpha(u)=m_K,\
\alpha(v)=m_H,\
u\equiv_q v
\end{array}
\right\}.
\]
Then
\[
\sepprof_{K,H}(n)
=
\min\{q:\defect_{q,n}^{K,H}=\varnothing\}.
\]
If $\sepprof_{K,H}$ is unbounded, then for every $q$ there exist
\[
u_q\in K,
\qquad
v_q\in H
\]
with
\[
u_q\equiv_q v_q.
\]
The regular-language single-profile proof gives a logarithmic lower bound
only when these witnesses can be chosen in a common contexted-power form:
\[
u_t=r x^{a+tp}s,
\qquad
v_t=r x^{b+tp}s.
\]
For arbitrary inseparable pairs, the existence of such synchronized witnesses
is not automatic.

\begin{problem}[Separator rate classification]\label{prob:separator-rate}
For regular disjoint $K,H$, classify the possible asymptotic rates of
\[
\sepprof_{K,H}(n)
\]
when $K$ and $H$ are not $\FO[<]$-separable. In particular, decide whether
every unbounded regular separator profile satisfies
\[
\sepprof_{K,H}(n)=\log_2 n+O_{K,H}(1),
\]
or whether slower unbounded rates can occur.
\end{problem}

A positive answer would require a separator analogue of
Lemma~\ref{lem:syntactic-extraction}: from nonseparability, one would need to
extract a common long-power scheme whose two sides land in $K$ and $H$
respectively. A negative answer would require a regular pair whose shortest
rank-$q$ inseparability witnesses have length superexponential in $q$.

\section{Conclusion}

The finite-horizon rank profile isolates one resource: the first-order
quantifier rank needed to classify a language on bounded-length observations.
For arbitrary languages, exact classification is always possible at
logarithmic rank. For all languages, boundedness of the profile is exactly
first-order definability. For regular languages, the profile has a sharp
algebraic dichotomy:
\[
\Syn(L)\text{ aperiodic}
\quad\Longleftrightarrow\quad
\rankprof_L(n)=O(1),
\]
and otherwise
\[
\rankprof_L(n)=\log_2 n+O_L(1).
\]
The proof of the unbounded case is syntactic and game-theoretic:
nonaperiodicity produces contexted periodic powers, and long powers of a
fixed word are indistinguishable by bounded-rank first-order sentences until
logarithmic depth.

The resulting picture is rigid for full $\FO[<]$ rank, but not for refined
resources. Rank-size Pareto profiles, fragment profiles, separator growth,
and optimized syntactic witnesses remain separate quantitative problems.

\end{document}